\documentclass[onecolumn,12pt]{IEEEtran}
\usepackage{epsfig,psfrag}
\usepackage{amsmath}
\usepackage{amssymb}
\usepackage{here}
\usepackage{dsfont}
\usepackage{mathrsfs}

\newtheorem{theorem}{Theorem}
\newtheorem{corollary}{Corollary}
\newtheorem{proposition}{Proposition}
\newtheorem{lemma}{Lemma}
\newtheorem{define}{Definition}
\newtheorem{example}{Example}

\newtheorem{remark}{Remark}
\newtheorem{alg}{Algorithm}

\newtheorem{problem}{Problem}

\DeclareMathOperator{\eps}{\varepsilon}
\DeclareMathOperator*{\col}{col}
\DeclareMathOperator*{\diag}{diag}
\DeclareMathOperator*{\esssup}{ess\ sup}
\DeclareMathOperator{\deltab}{\boldsymbol\delta}
\DeclareMathOperator{\Sigmab}{\boldsymbol\Sigma}
\title{Robust stability and stabilization of uncertain linear positive systems via Integral Linear Constraints:\\ $L_1$- and $ L_\infty$-gains characterization}

\author{Corentin Briat\thanks{The author was previously with the ACCESS Linnaeus Centre, KTH, Stockholm, Sweden. He is now with the Swiss Federal Institute of Technology--Zurich (ETHZ), Department of Biosystems Science and Engineering (D-BSSE), Mattenstrasse 26, 4058 Basel, Switzerland; email: corentin@briat.info, corentin.briat@bsse.ethz.ch; url: http://www.briat.info}}

\begin{document}





\maketitle

\begin{abstract}
Copositive linear Lyapunov functions are used along with dissipativity theory for stability analysis and control of uncertain linear positive systems. Unlike usual results on linear systems, linear supply-rates are employed here for robustness and performance analysis using $L_1$- and $ L_\infty$-gains. Robust stability analysis is performed using Integral Linear Constraints (ILCs) for which several classes of uncertainties are discussed. The approach is then extended to robust stabilization and performance optimization. The obtained results are expressed in terms of robust linear programming problems that are equivalently turned into finite dimensional ones using Handelman's Theorem. Several examples are provided for illustration.
\end{abstract}


\section{Introduction}

Linear internally positive systems are a particular class of linear systems whose state takes only nonnegative values. Such models can represent many real world processes, from biology \cite{Haddad:10}, passing through ecology and epidemiology \cite{Murray:02}, to networking \cite{Shorten:06}. Compartmental models, used e.g. in biological, medical, epidemiological applications, are also generally expressed as (non)linear nonnegative systems \cite{Haddad:10}. 
Several works have been devoted to their analysis and control, see e.g. \cite{Farina:00, Leenheer:01, Haddad:05, Aitrami:07, Bru:09, Shorten:09, Tanaka:10, Tanaka:11, Ebihara:11,Rantzer:11}.

\emph{Quadratic Lyapunov functions} of the form $V(x)=x^TPx$, with $P=P^T$ positive definite, are the most commonly used to study the stability of linear systems. Dissipativity theory with \emph{quadratic storage functions} and \emph{quadratic supply-rates} are also widely used for robustness analysis, e.g. through the full-block S-procedure \cite{Scherer:97}, and norms computation, e.g. the $\mathcal{H}_\infty$-norm and generalized $\mathcal{H}_2$-norm. This quadratic framework also allows to apply powerful analysis techniques such as those based on \emph{Integral Quadratic Constraints} (IQCs) \cite{RantzerMegretski:97} (exploiting the KYP-Lemma \cite{Kalman:63,Rantzer:96,Scherer:05a} and the Plancherel Theorem), and \emph{well-posedness theory} \cite{Safonov:80, Iwasaki:98, Laas:07} which is mainly based on topological separation \cite{Safonov:80}. An interesting point in these approaches is that many of the obtained results can be represented as optimization problems involving Linear Matrix Inequalities (LMIs) \cite{Boyd:94a}, a wide class of convex optimization problems solvable in polynomial-time using e.g. interior-point algorithms \cite{NestNemi:94}.


Unlike general linear systems, the stability of linear positive systems can be losslessly analyzed by considering a positive definite \emph{diagonal matrix} $P$. Such functions are then called \emph{diagonal Lyapunov functions} and the corresponding stability notion \emph{diagonal stability} \cite{Barker:78}. This particular form dramatically simplifies the structure of the stability and stabilization problems, allowing for instance the design of structured and decentralized controllers in a simple way (convex formulation). Convexity does not hold in the general linear systems case where the design of structured controllers is known to be a NP-hard problem \cite{Blondel:97}. Recently, a KYP-Lemma for positive systems using diagonal Lyapunov functions has been obtained in \cite{Tanaka:10} by using the results reported in \cite{Shorten:09}. In the same vein, exact losslessness conditions for robust stability analysis have been further obtained in \cite{Tanaka:11} where several classes of uncertainties are covered.

There exists however another class of Lyapunov functions leading to necessary and sufficient conditions. They are referred to as \emph{linear copositive Lyapunov functions} \cite{Haddad:05, Aitrami:07, Mason:07, AitRami:09, Knorn:09, Fornasini:10} and write $V(x)=\lambda^Tx$, where $\lambda$ is a vector with positive entries. In such a case, the resulting stability condition can be expressed as a linear programming problem (convex again) and solved in an efficient way. Since the Lyapunov function is linear, there is no more relationship with the vector 2-norm and the $L_2$-norm as in the quadratic case, but rather with the vector 1-norm and the $L_1$-norm. This framework is then more suitable for the analysis of the $L_1$-gain of positive systems and its consideration in robustness and performance characterization.

%

In this paper, the stability analysis and control of uncertain positive systems is considered in the $L_1$-induced norm and $L_\infty$-induced norm using linear copositive Lyapunov functions and dissipativity theory \cite{Willems:72,Haddad:05} with linear supply-rates. Stability analysis and control synthesis results for unperturbed systems are first provided to set up the ideas and introduce the important tools. It is shown that computing the $L_1$- and $L_\infty$-gains for linear positive systems is tantamount to solving a linear programming problem, with a complexity growing linearly with respect to the system size. The $L_1$-gain is determined via a direct application of dissipativity theory while the $L_\infty$-gain is computed as the $L_1$-gain of the transposed of the original system. It is also shown that the computed gains are valid regardless of the sign of inputs and states, relaxing then the concept of positive system to certain conditions on the system matrices only. The consideration of nonnegative states and inputs are for computational and theoretical considerations only, i.e. the use of linear Lyapunov functions and linear supply-rates. Convex necessary and sufficient stabilization conditions using full, structured and bounded state-feedback controllers and accounting for performance bounds are then obtained. While the question of determining under which algebraic conditions a system can be made positive in closed-loop is still open, the proposed methodology can be applied to implicitly characterize them, together with a performance constraint.

Robust stability analysis is performed using Linear Fractional Transformations (LFTs), a classical tool of robust analysis \cite{Zhou:96}, very few used in the context of positive systems \cite{Briat:11g,Tanaka:11,Ebihara:11}. The advantage of using LFT is that parameter-varying systems depending rationally and polynomially in the parameters can be expressed as a simple interconnection of a time-invariant linear system and an uncertain matrix depending linearly on the parameters, simplifying then the analysis. The use of LFT also extends to any type of uncertainties and nonlinearities, e.g. delay operators, uncertain stable transfer functions, more general bounded uncertainties and static nonlinearities. In the context of uncertain positive systems, the overall system is rewritten as a \emph{positive interconnection} of a nominal system and a matrix of uncertain positive operators. These uncertain operators are characterized through Integral Linear Constraints (ILCs), the linear counterpart of IQCs. Although the provided framework does not enjoy the availability of the KYP Lemma nor the Plancherel Theorem, a frequency domain analysis can still be used in order to select the scalings accurately. Several classes of uncertainties are discussed and it is shown that for linear time-invariant uncertainties, ILCs fully characterize their static-gain matrices. Based on this fact, several general robustness properties are discussed.

Robust stability analysis results are finally derived using dissipativity theory and formulated as robust linear programming problems. Exact stability conditions are provided in the particular case of LTI positive uncertainties with fixed static-gain matrix. To improve tractability, Handelman's Theorem \cite{Handelman:88} is used to produce linear programs involving a finite number of constraints. A procedure to reduce the number of extra variables introduced by Handelman's Theorem is also proposed to reduce the computational complexity of the approach. The results are finally extended to robust stabilization. It is shown that the presence of scalings does not destroy the convexity of the problem as it is often the case in full-block S-procedure \cite{Scherer:97} or IQC approaches. Several examples demonstrate the efficiency of the approach and its exactness for a certain class of uncertainties including delays.

\textbf{Outline:} Section \ref{sec:prel} introduces the problem, fundamental definitions and results. Section \ref{sec:unp} is devoted to the stability analysis of unperturbed systems and Section \ref{sec:stabiunp} deals with their stabilization. This is extended to robust stability analysis and robust stabilization in Sections \ref{sec:per} and \ref{sec:stabip} respectively. In Section \ref{sec:solvingunc}, Handelman's Theorem is used to relax the robust optimization problems formulated in the previous sections. The results are then finally illustrated through examples in Section \ref{sec:ex}.

\textbf{Notations:}  $\mathds{1}_n\in\mathbb{R}^n$ denotes the column vector containing entries equal to 1. For general real matrices or vectors $A,B\in\mathbb{R}^{n\times m}$, the inequality $A<(\le)B$ is componentwise. Let $x\in\mathbb{R}^n$, the vector $\sigma$-norm, $\sigma$ positive integer, $\sigma<\infty$, is denoted by $||x||_\sigma=\left(\sum_{i=1}^n|x_i|^\sigma\right)^{1/\sigma}$, while the vector $\infty$-norm is defined by $||x||_\infty=\max_{i\in\{1,\ldots,n\}}|x_i|$. Given $v:[0,\infty)\to\mathbb{R}^n$, the $L_\sigma$-norm $||v||_{ L_\sigma}$, $\sigma$ positive integer, $\sigma<\infty$, and the $L_\infty$-norm $||v||_{ L_\infty}$ are defined by
$||v||_{ L_\sigma}=\left(\int_{0}^\infty||v(t)||_\sigma^\sigma dt\right)^{1/\sigma}$ and $||v||_{ L_\infty}=\esssup_{t\ge0}||v(t)||_\infty$, respectively. The spaces of signals $v:[0,\infty)\to\mathbb{R}^n$ having finite $L_\sigma$-norm is denoted by $L_\sigma^n$. When the dimension of the signal has no importance, we will use the shorthand $L_\sigma$.
For a matrix $X\in\mathbb{R}^{n\times m}$, $[X]_{ij}$, $[X]_{r,i}$ and $[X]_{c,i}$ denote the $(i,j)$ scalar entry, the $i^{th}$ row and the $i^{th}$ column. A linear map $\ell(x)=c^Tx$ is said to be copositive if $c^Tx\ge0$ for all $x\ge0$.
We also define the sets $\mathbb{R}^n_{++}=\left\{\alpha\in\mathbb{R}^n: \alpha>0\right\}$, $\mathbb{R}^n_+:=\left\{\alpha\in\mathbb{R}^n:\ \alpha\ge0, ||\alpha||\ne0\right\}$ and $\bar{\mathbb{R}}^n_+:=\left\{\alpha\in\mathbb{R}^n: \alpha\ge0\right\}$.

\section{Preliminaries}\label{sec:prel}

Let us consider general LTI systems of the form:
\begin{equation}\label{eq:mainsyst}
  \begin{array}{lcl}
    \dot{x}(t)&=&Ax(t)+Bu(t)+Ew(t)\\
    z(t)&=&Cx(t)+Du(t)+Fw(t)\\
    x(0)&=&x_0
  \end{array}
\end{equation}
where $x,x_0\in\mathbb{R}^n$, $u\in\mathbb{R}^m$, $w\in\mathbb{R}^p$ and $z\in\mathbb{R}^q$ are respectively the system state, the initial condition, the control input, the exogenous input and the controlled output. When the system (\ref{eq:mainsyst}) with $u\equiv0$ is asymptotically stable, it defines an operator from $L_\sigma\ni w\to z\in L_\sigma$, $\sigma$ positive integer. Such an operator framework is suitable for defining, computing and optimizing norms of systems \cite{Desoer:75a}, and determining robustness and performance properties.

%
\begin{define}\label{def:pos}
  Let us consider the uncontrolled version of system (\ref{eq:mainsyst}), i.e. $u\equiv 0$. The system is said to be \emph{positive} if the following conditions hold:
\begin{enumerate}
  \item[i)] The matrix $A$ is Metlzer, i.e. it has nonnegative off-diagonal entries,
  \item[ii)] The matrices $E,C,F$ are nonnegative, i.e. they only have nonnegative entries.
\end{enumerate}
\end{define}
A consequence of the above definition is that a) when $w\equiv 0$, we have $x(t)\in\bar{\mathbb{R}}^n_+$ for all $t\ge0$ and all $x_0\in\bar{\mathbb{R}}^n_+$; b) when  $x_0=0$, we have $z(t)\in\bar{\mathbb{R}}^q_+$ for all $t\ge0$ and all $w(t)\in\bar{\mathbb{R}}^p_+$. These facts justify the denomination of positive system. Note that the terminology also degenerates to the case of autonomous positive systems for which only the matrix $A$ must be Metzler.
%

\begin{define}
  A linear map $V(x)=\lambda^Tx$ with $V(0)=0$ is said to be a linear copositive Lyapunov function for the positive system $\dot{x}(t)=Ax(t)$ if both $V(x)>0$ and $\dot{V}(x)<0$ hold for all  $x\in\mathbb{R}^n_+$.
\end{define}




\begin{define}[$L_\sigma$-gains of operators]
Given an operator $\Sigma: L_\sigma^p\mapsto L_\sigma^q$, $\sigma$ positive integer, the $L_\sigma$-gain $||\Sigma||_{L_\sigma-L_\sigma}$ is defined as $$||\Sigma||_{L_\sigma-L_\sigma}:=\sup_{||w||_{L_\sigma}=1}||\Sigma w||_{L_\sigma}.$$
Equivalently, it is the smallest $\theta\ge0$ such that
$$||\Sigma w||_{L_\sigma}\le\theta||w||_{L_\sigma}$$
holds for all $w\in L_\sigma$.
%
\end{define}

\begin{define}\label{def:Lnorms}
  The $L_1$-gain  of an asymptotically stable linear time-invariant system $H$ with transfer function $\widehat{h}(s)=C(sI-A)^{-1}E+F$ mapping $p$ inputs to $q$ outputs is given by \cite{Desoer:75a}:
\begin{equation}\label{eq:L1def}
||H||_{L_1-L_1}=\max_{j\in\{1,\ldots,q\}}\left\{\sum_{i=1}^p\int_0^{+\infty}|h_{ij}(t)|dt\right\}
\end{equation}
where $h_{ij}(t)$ is the impulse response from input $j$ to output $i$. In the same way, the $L_\infty$-gain is given by
\begin{equation}
||H||_{ L_\infty- L_\infty}=\max_{i\in\{1,\ldots,p\}}\left\{\sum_{j=1}^q\int_0^{+\infty}|h_{ij}(t)|dt\right\}.
\end{equation}
\end{define}
The $L_1$-gain quantifies the gain of the \textit{most influent input} since the max is taken over the columns. In contrast, the $ L_\infty$-gain of a system is the max taken over the rows and then characterizes the \textit{most sensitive output}. Note that in the SISO case, the 2 norms obviously coincide, and so do all the $L_\sigma$-norms, as pointed out in \cite{Rantzer:11}. Another important fact, needed later, is the correspondence between the $L_1$-induced and $L_\infty$-induced norms using the notion of transposed system:
\begin{proposition}\label{prop:trans}
Let us consider a system $H$ with transfer function $\widehat{h}(s)=C(sI-A)^{-1}E+F$ and its corresponding transposed system $H^*$ having transfer function  $\widehat{h}^*(s)=E^T(sI-A^T)C^T+F^T$. The $L_\infty$-gain of  $H$ is related to the $L_1$-gain of  $H^*$ through the equality:
  \begin{equation}
  ||H||_{L_\infty-L_\infty}=||H^*||_{L_1-L_1}.
\end{equation}
%
%
\end{proposition}
\begin{proof}
  The proof follows from the definitions of the transposed system and the norms.
\end{proof}

%

\begin{proposition}\label{prop:normp}
  Given an asymptotically stable positive system $H$ with impulse response $h(t)\in\mathbb{R}^{q\times p}$, $t\ge0$, we have
\begin{equation}
||H||_{L_1-L_1}=\max\left\{\mathds{1}_q^T\widehat{h}(0)\right\}\quad \mathrm{and}\quad ||H||_{L_\infty- L_\infty}=\max\left\{\widehat{h}(0)\mathds{1}_p\right\}
\end{equation}
where $\widehat{h}$ is the Laplace transform of $h$.
\end{proposition}
\begin{proof}
Since $A$ is Metlzer and Hurwitz, then we have $e^{At}\ge0$ for all $t\ge0$. It is then immediate to see that the impulse response $h(t)=Ce^{At}E+F$ is nonnegative as well. From the asymptotic stability of the system and (\ref{eq:L1def}) we hence have
\begin{equation*}
\begin{array}{lcl}
    \int_0^{+\infty}h(t)dt  &=& \left.\int_0^{+\infty}h(t)e^{-st}dt\right|_{s=0}\\
                            &=& \widehat{h}(0)
\end{array}
\end{equation*}
and we get the result for the $L_1$-gain. The $L_\infty$-gain expression follows from Proposition \ref{prop:trans}.
\end{proof}
The above results then show that the $L_1$- and $L_\infty$-gains of positive systems are intimately related to each others, and to the static-gain matrix of the system. The formulas given in Proposition \ref{prop:normp} also suggest that a theoretical analysis of the gains is possible.

\section{Stability and Performance Analysis of Unperturbed systems}\label{sec:unp}

In this section, nonconservative stability and performance analysis criteria for unperturbed systems are derived. It is assumed throughout this section that the system (\ref{eq:mainsyst}) is positive and that the control input is identically zero, i.e. $u\equiv0$. Similar results can be found in \cite{Zappavigna:10a, Briat:11g, Ebihara:11}.

\subsection{$L_1$-gain characterization and computation}

\begin{lemma}[$L_1$-gain characterization]\label{lem:L1}
  Let us consider system (\ref{eq:mainsyst}) and assume it is positive. Then, the following statements are equivalent:
\begin{enumerate}
\item[i)] System (\ref{eq:mainsyst}) is asymptotically stable and the $L_1$-gain of the transfer $w\mapsto z$ is smaller than $\gamma$.
  \item[ii)] System (\ref{eq:mainsyst}) is asymptotically stable and $\mathds{1}^T_q\widehat{h}(0)<\gamma\mathds{1}^T_p$ where $\widehat{h}(s)$ is the transfer function associated to system (\ref{eq:mainsyst}).
  \item[iii)] System (\ref{eq:mainsyst}) is asymptotically stable and $\mathds{1}^T_q(F-CA^{-1}E)<\gamma\mathds{1}^T_p$.
  \item[iv)] There exists $\lambda\in\mathbb{R}^n_{++}$ such that the linear program
  \begin{subequations}\label{eq:L1stability}
  \begin{gather}
     \lambda^TA+\mathds{1}_{q}^TC<0\label{eq:L1stability1}\\
     \lambda^TE-\gamma\mathds{1}_{p}^T+\mathds{1}_{q}^TF<0\label{eq:L1stability2}
  \end{gather}
\end{subequations}
 is feasible.
\end{enumerate}
\end{lemma}
\begin{proof}
It is immediate that ii) and iii) are equivalent since $\widehat{h}(0)=F-CA^{-1}E$. We then prove that iv) and i) are equivalent and, finally, consider the equivalence of iii) and iv).

\textbf{Proof of iv)$\Rightarrow$i)}

The proof relies on dissipativity theory for nonnegative systems \cite{Willems:72,Haddad:05}. Let us consider the copositive linear storage function ${V(x)=\lambda^Tx}$ with ${\lambda\in\mathbb{R}^n_{++}}$ and the supply rate ${s(w,z)=\gamma||w||_1-||z||_1}$, $\gamma>0$. Then, according to dissipativity theory, if the functional
  \begin{equation}
    \mathcal{H}(x,w,z)=V(x(t))-\int_0^ts(w(\eta),z(\eta))d\eta.
  \end{equation}
  is decreasing along the trajectories solution of the system (\ref{eq:mainsyst}), then the system is dissipative with respect to the supply-rate $s(w,z)$ and the $L_1$-gain is bounded from above by $\gamma$. Moreover, since $\lambda>0$ holds, then asymptotic stability of the system also follows. The derivative of $\mathcal{H}$ along the trajectories of (\ref{eq:mainsyst}) is given by
  \begin{equation*}
    \begin{array}{lcl}
      \dot{\mathcal{H}}&=&\lambda^T\dot{x}(t)-\gamma\mathds{1}_p^Tw(t)+\mathds{1}_q^Tz(t)\\
                        &=&\begin{bmatrix}
                        \lambda^TA+\mathds{1}_q^TC &\ \ & \lambda^TE-\gamma\mathds{1}_p^T+\mathds{1}_q^TF
  \end{bmatrix}\begin{bmatrix}
    x(t)\\
    w(t)
  \end{bmatrix}.
    \end{array}
  \end{equation*}
  Since the signals $x(t)$ and $w(t)$ are nonnegative, then $\dot{\mathcal{H}}$ is negative on $\mathbb{R}_+^{n+p}$ if and only if the left factor is negative, or equivalently if the conditions (\ref{eq:L1stability}) hold.

\textbf{Proof of i)$\Rightarrow$iv)}

The necessity comes from Theorems 5.1, 5.3 and 6.2 of \cite{Haddad:05}. This can also be viewed from the fact that the S-procedure is lossless for any number of constraints in the linear case \cite{Yaku:77, Jonsson:01}. An alternative proof given in \cite{Ebihara:11} relies on Farkas' Lemma (which is actually an alternative way of seeing the linear S-procedure \cite[Remark 4]{ Jonsson:01}).

\textbf{Proof of iv)$\Rightarrow$iii)}

A similar proof is given in \cite{Ebihara:11}. Assume $4)$ holds, then the matrix $A$ is Metlzer and Hurwitz, implying in turn that $A^{-1}$ is a nonpositive matrix, i.e. $A^{-1}\le 0$. By right-multiplying inequality (\ref{eq:L1stability1}) by $A^{-1}$, we obtain $  \lambda^T>-\mathds{1}^T_qCA^{-1}$ and after substitution into (\ref{eq:L1stability2}), we get the inequality $  \mathds{1}^T_q(F-CA^{-1}E)<\gamma\mathds{1}_p^T$ which coincides with the one of statement iii).

\textbf{Proof of iii)$\Rightarrow$iv)}
The proof relies on an explicit construction of a Lyapunov function such that asymptotic stability of the system implies the feasibility of the linear program of statement iv). Assume iii) holds, then there exist $\nu\in\mathbb{R}_{++}^n$ and $\eps>0$ such that $\nu^TA<0$ and
\begin{equation}\label{eq:jenlole1}
  \mathds{1}^T_qF+(\eps\nu^T-\mathds{1}^T_qCA^{-1})E<\gamma\mathds{1}_p^T.
\end{equation}
Noting that the term inside parentheses is positive since $\nu>0$ and $A^{-1}\le0$, then we can let $\lambda^T:=\eps\nu^T-\mathds{1}^T_qCA^{-1}$. By right-multiplying this inequality by $A$, we obtain
\begin{equation}\label{eq:jenlole2}
  \lambda^TA+\mathds{1}^T_qC=\eps\nu^TA
\end{equation}
which is negative from assumptions $\nu^TA<0$ and $\eps>0$. Since the left-hand term of (\ref{eq:jenlole2}) is identical to the one of (\ref{eq:L1stability1}), the implication of the feasibility of (\ref{eq:L1stability1}) is proved. Reorganizing now the terms of the inequality (\ref{eq:jenlole1}), we get $\lambda^TE-\gamma\mathds{1}_p^T+\mathds{1}^T_qF<0$, which is identical to (\ref{eq:L1stability2}). The proof is complete.
\end{proof}

\begin{remark}
  It is interesting to point out that despite of being computed with the assumption of nonnegative input signals $w\in L_1$ and nonnegative state values, \emph{the determined $L_1$-gain is valid for any input signal in $L_1$ and any initial state $x_0\in\mathbb{R}^n$}. This is due to the fact that the $L_1$-gain has alternative definition (\ref{eq:L1def}), which depends on the nonnegativity of impulse response only. Nonnegative input and output signals are considered for theoretical and computational convenience only.
\end{remark}

Lemma \ref{lem:L1} can be used to compute the exact $L_1$-gain of any asymptotically stable positive linear system as follows:
\begin{alg}[Computing the $L_1$-gain]\label{alg:L1stab}
The gain coincides with the optimal value of the following linear programming problem:
  \begin{equation*}
  \begin{array}{lcl}
        \min_{\lambda,\gamma}\ \gamma & \mathrm{s.t.}  & \lambda\in\mathbb{R}_{++}^n,\ \gamma>0\ \mathrm{and}\ (\ref{eq:L1stability})\ \mathrm{hold.}
  \end{array}
  \end{equation*}
\end{alg}

The complexity of the above linear programming problem is quite low since it is a linear program involving $n+1$ decision variables and $2n+p+1$ constraints. The complexity then grows linearly with respect to the size of the system.

\subsection{$L_\infty$-gain characterization and computation}

The following result is the $L_\infty$ counterpart of Lemma \ref{lem:L1}.
\begin{lemma}[$L_\infty$-gain characterization]\label{lem:Linf}
  Let us consider system (\ref{eq:mainsyst}) that we assume to be positive. Then, the following statements are equivalent:
\begin{enumerate}
\item[i)] System (\ref{eq:mainsyst}) is asymptotically stable and the $L_\infty$-gain of the transfer $w\mapsto z$ is smaller than $\gamma$.
  \item[ii)] System (\ref{eq:mainsyst}) is asymptotically stable and $\widehat{h}(0)\mathds{1}_p<\gamma\mathds{1}_q$ where $\widehat{h}(s)$ is the transfer function associated to system (\ref{eq:mainsyst}).
  \item[iii)] System (\ref{eq:mainsyst}) is asymptotically stable and $(F-CA^{-1}E)\mathds{1}_p<\gamma\mathds{1}_q$.
  \item[iv)] There exists $\lambda\in\mathbb{R}^n_{++}$ such that the linear program
  \begin{subequations}\label{eq:Linfstability}
  \begin{gather}
     A\lambda+E\mathds{1}_{p}<0\label{eq:Linfstability1}\\
      C\lambda-\gamma\mathds{1}_{q}+F\mathds{1}_{p}<0\label{eq:Linfstability2}
  \end{gather}
\end{subequations}
 is feasible.
\end{enumerate}
\end{lemma}
\begin{proof}
  The proof relies on Proposition \ref{prop:trans}. By substituting the matrices of the transposed system into the conditions (\ref{eq:L1stability}) of Lemma \ref{lem:L1}, we get the conditions (\ref{eq:Linfstability}). The equivalence between the statements is immediate from Lemma \ref{lem:L1}.
%
\end{proof}

Similarly as for the $L_1$-gain, it is possible to compute the $L_\infty$-gain through an optimization problem.
\begin{alg}[Computing the $L_\infty$-gain]\label{alg:Linfstab}
The gain coincides with the optimal value of the following linear programming problem:
  \begin{equation*}
  \begin{array}{lcl}
        \min_{\lambda,\gamma}\ \gamma &\mathrm{s.t.}  & \lambda\in\mathbb{R}_{++}^n,\ \gamma>0\ \mathrm{and}\ (\ref{eq:Linfstability})\ \mathrm{hold.}
  \end{array}
  \end{equation*}
\end{alg}

Unlike the $L_1$-gain case, the number of constraints is $2n+q+1$ while the number of variables remains the same. Again, the computational complexity grows linearly with respect to the size of the system.

\section{Stabilization of unperturbed systems}\label{sec:stabiunp}

This section is devoted to the stabilization of unperturbed linear systems via state-feedback control laws of the form:
\begin{equation}\label{eq:cl}
  u(t)=Kx(t)
\end{equation}
where the controller gain $K$ belongs to one of the following sets:
\begin{enumerate}
  \item The set of unconstrained controllers $\mathcal{K}:=\mathbb{R}^{m\times n}$.
  \item The set of structured controllers $\mathcal{K}_c:=\left\{K\in\mathbb{R}^{m\times n}:\ [K]_{ij}=0,\ (i,j)\in\mathcal{S}_c\right\}$, where $\mathcal{S}_c$ is the set of indices corresponding to 0 entries in the controller gain.
  \item The set of bounded controllers $\mathcal{K}_b:=\left\{K\in\mathbb{R}^{m\times n}:\ K^-\le K\le K^+\right\}$, where $K^-$ and $K^+$ are the lower and upper bounds on the controller gain $K$, respectively.
\end{enumerate}
Unlike general LTI systems, the design of structured and bounded controllers is not a NP-hard problem \cite{Blondel:97,Peaucelle:08, Briat:09c}. By indeed using a diagonal Lyapunov function, necessary and sufficient conditions for the design structured controllers can be easily obtained. We show here that it is also the case in the current linear setting and that it readily extends to the stabilization of positive systems with guaranteed $L_\infty$-gain.

%
%
%
\begin{lemma}[Stabilization with $K\in\mathcal{K}$]\label{th:stabunp}
  Let us consider the closed-loop system (\ref{eq:mainsyst})-(\ref{eq:cl}), where (\ref{eq:mainsyst}) is not necessarily a positive system, with transfer function $\widehat{h}_{cl}(s,K):=(C+DK)(sI-A-BK)^{-1}E+F$. Then, the following statements are equivalent:
\begin{enumerate}
  \item[i)] There exists a controller matrix $K$ such that the closed-loop system (\ref{eq:mainsyst})-(\ref{eq:cl}) is positive, asymptotically stable and the $L_\infty$-gain of the transfer $w\mapsto z$ is less than $\gamma>0$.
  \item[ii)] There exists a controller matrix $K$ such that the closed-loop system (\ref{eq:mainsyst})-(\ref{eq:cl}) is positive, asymptotically stable and verifies ${\widehat{h}_{cl}(0,K)\mathds{1}_p<\gamma\mathds{1}_p}$.
  \item[iii)] There exist $\lambda\in\mathbb{R}^n_{++}$ and $\mu_i\in\mathbb{R}^m$, $i=1,\ldots,n$ such that the linear program
      \begin{subequations}\label{eq:Linfstabz}
  \begin{gather}
    A\lambda+B\sum_{i=1}^n\mu_i+E\mathds{1}_{p}<0\label{eq:Linfstabz1}\\
    C\lambda+D\sum_{i=1}^n\mu_i-\gamma\mathds{1}_{q}+F\mathds{1}_{p}<0\label{eq:Linfstabz2}\\
      \left[A\right]_{ij}\lambda_j+[B]_{r,i}\mu_j\ge0,\ i,j=1,\ldots, n,\ i\ne j\label{eq:Linfstabz3}\\
      \left[C\right]_{ij}\lambda_j+[D]_{r,i}\mu_j\ge0,\ i=1,\ldots,q,\ j=1,\ldots, n\label{eq:Linfstabz4}
  \end{gather}
\end{subequations}
 is feasible. In such a case, a suitable $K$ is given by
 \begin{equation}
         K=\begin{bmatrix}
           \lambda_1^{-1}\mu_1 & \ldots & \lambda_n^{-1}\mu_n
         \end{bmatrix}.
       \end{equation}
\end{enumerate}
\end{lemma}
\begin{proof}
To prove the equivalence it is enough to show that statements 1) and 3) are equivalent. The rest follows from the proofs of Lemmas \ref{lem:L1} and \ref{lem:Linf}. The closed-loop system is given by
  \begin{equation}
    \begin{array}{lcl}
      \dot{x}(t)&=&(A+BK)x(t)+Ew(t),\\
      z(t)&=&(C+DK)x(t)+Fw(t).
    \end{array}
  \end{equation}
 Substituting the closed-loop system into (\ref{eq:Linfstability}) yields
  \begin{equation}
    \begin{bmatrix}
      (A+BK)\lambda+E\mathds{1}_{p}\\
      (C+DK)\lambda-\gamma\mathds{1}_{q}+F\mathds{1}_{p}
    \end{bmatrix}<0.
  \end{equation}
Noting that $K\lambda=\sum_{i=1}^n\lambda_i[K]_{c,i}$, it turns out that the change of variable $\mu_i=[K]_{c,i}\lambda_i$ linearizes the problem and yields (\ref{eq:Linfstabz1})-(\ref{eq:Linfstabz2}). To ensure the positivity constraint on the closed-loop system we need to impose $A+BK$ and $C+DK$ to be Metzler and nonnegative, respectively. Using the same procedure as in \cite{Aitrami:07}, these constraints are captured by inequalities (\ref{eq:Linfstabz3})-(\ref{eq:Linfstabz4}).
\end{proof}
The above theorem can hence be viewed as an extension of \cite{Aitrami:07} where no performance criterion is considered. The following proposition discusses the cases $K\in\mathcal{K}_c$ and $K\in\mathcal{K}_b$.
\begin{proposition}[Cases $K\in\mathcal{K}_c$ and $K\in\mathcal{K}_b$]
These cases can be easily handled by adding supplementary constraints to the linear program of Lemma \ref{th:stabunp}.
\begin{itemize}
  \item The design of a controller gain $K\in\mathcal{K}_c$ is ensured by considering the additional linear constraints ${\begin{bmatrix}
      \mu_1 & \ldots & \mu_m
    \end{bmatrix}_{ij}=0}$, for all ${(i,j)\in\mathcal{S}_c}$.
  \item Prescribed bounds on the coefficients of the controller gain are imposed, i.e. $K\in\mathcal{K}_b$ by considering the additional linear constraints ${[K^-]_{r,i}\lambda_i\le\mu_i\le[K^+]_{r,i}\lambda_i}$, $i=1,\ldots, m$, $K^-,K^+\in\mathbb{R}^{m\times n}$, $K^-\le K^+$.
\end{itemize}
\end{proposition}

Note that in both cases the necessity of the approach is preserved. The above results are thus nonconservative. Other constraints like asymmetric bounds on the control input and the consideration of bounded states can also be easily considered \cite{Aitrami:07}.


\section{Robust stability analysis and robust performance}\label{sec:per}

Let us focus now on uncertain linear systems subject to real parametric uncertainties ${\delta\in\deltab:=[0,1]^N}$, $N>0$, of the form
\begin{equation}\label{eq:unsyst}
\begin{array}{lcl}
    \dot{x}(t)&=&A_\delta(\delta)x(t)+B_\delta(\delta)u(t)+E_\delta(\delta)w_1(t)\\
    z_1(t)&=&C_\delta(\delta)x(t)+D_\delta(\delta)u(t)+F_\delta(\delta)w_1(t)\\
    x(0)&=&x_0
\end{array}
\end{equation}
where $x\in\bar{\mathbb{R}}^n_+$, $x_0\in\bar{\mathbb{R}}^n_+$, $u\in\mathbb{R}^m$, $w_1\in\bar{\mathbb{R}}^p_+$ and $z_1\in\bar{\mathbb{R}}^q_+$ are the system state, the initial condition, the control input, the exogenous input and the performance output respectively. We assume in this section that the above system is positive, that is, for all $\delta\in\deltab$, the matrix $A_\delta(\delta)$ is Metzler, and $E_\delta(\delta),C_\delta(\delta),F_\delta(\delta)$ are nonnegative matrices. We also assume that the system matrices are continuous functions of $\delta$.

While this section mainly focuses on positive systems subject to real parametric uncertainties, the proposed methodology basically applies to any type of positive interconnections involving, for instance, delays, positive infinite-dimensional operators, time-invariant and time-varying uncertain positive operators, static sign-preserving nonlinearities, etc. A justification for focusing on parametric uncertainties lies in the particular property of positive systems that (as we shall see later in Section \ref{sec:handling}) only the static-gain matrix is critical for evaluating stability of interconnections. This property implies that many problems involving uncertain positive transfer functions equivalently reduce to problems involving constant parametric uncertainties.

%
%

\subsection{Positive Linear Fractional Transformation}

Using LFT, the system (\ref{eq:unsyst}) is rewritten as
\begin{equation}\label{eq:LFR}
  \begin{array}{lcl}
    \dot{x}(t)&=&Ax(t)+E_0w_0(t)+E_1w_1(t)\\
    z_0(t)&=&C_0x(t)+F_{00}w_0(t)+F_{01}w_1(t)\\
    z_1(t)&=&C_1x(t)+F_{10}w_0(t)+F_{11}w_1(t)\\
    w_0(t)&=&\Delta(\delta)z_0(t)
  \end{array}
\end{equation}
where the loops signals $z_0(t),w_0(t)\in\bar{\mathbb{R}}^{n_0}_+$ have been added. It is important to keep in mind that we are considering robustness in the $L_1$-norm, the tractability of which relying on the nonnegativity of the loop signals $w_0$ and $z_0$. Representation (\ref{eq:LFR}) is not unique, even when minimal, so the arising question concerns the constructability of a so-called \emph{positive Linear Fractional Representation}(LFR), i.e. an LFR of the form (\ref{eq:LFR}) with nonnegative loop signals and positive operators, corresponding to the positive system (\ref{eq:unsyst}).  Positive interconnections are recurrent in the analysis of positive systems and have already been considered in several works \cite{Farina:00, AitRami:06, Ebihara:11, Briat:11g}. By picking suitable matrices for the LFR, it is always possible to make the loop signals nonnegative. 
%
%
%

\subsection{Handling positive uncertainties via Integral Linear Constraints}\label{sec:handling}

In robust stability analysis theories addressing two interconnected systems, stability conditions very often consist of two separate conditions: one for each subsystem. The first one, generally very precise, is used to characterize the nominal system which usually enjoys nice properties like linearity, time-invariance, etc. The second condition, often more difficult to derive when high precision is sought, is used to characterize the uncertain part. Several approaches have been developed to study the stability of interconnections like small-gain results \cite{Zhou:96} and generalizations \cite{Scherer:97}, well-posedness/quadratic separation \cite{Iwasaki:98a,Laas:07} (based on topological separation arguments \cite{Safonov:80}) and IQCs \cite{RantzerMegretski:97}. IQCs are very powerful objects that are able to implicitly describe, with high accuracy, uncertain operators through the characterization of their input/output signals. They unfortunately do not fit in the current framework since they are based on quadratic forms, while the proposed approach considers linear ones. Inspired from this idea, Integral Linear Constraints (ILCs) are considered in this paper.

\subsubsection{Definitions of ILCs and preliminary results}\ \\

\vspace{-3mm}\noindent Let us consider first an uncertain operator $\Sigma:L_1\mapsto L_1$ belonging to a known set $\Sigmab$. The main idea is to implicitly characterize the set $\Sigmab$ using ILCs as\footnote{the absolute value is componentwise.}
\begin{equation}
  \Sigmab\subseteqq\left\{\Sigma:\int_0^{+\infty} \varphi_{1,i}^T\left(|w(t)|+\varphi_{2,i}^T|z(t)|\right)dt\ge0,\ z=\Sigma w,\ i=1,\ldots\right\}.
\end{equation}
The vectors $\varphi_{1,i}$ and $\varphi_{2,i}$ are called \emph{scalings} and must be chosen according to the set $\Sigmab$. Note that the absolute values disappear when considering a set of positive operators and nonnegative input signals.

One important feature of IQCs lies in the fact that, by virtue of the Plancherel Theorem, it is possible to express the inequality in the frequency domain, domain in which the tuning of the scalings may be easier and/or more precise. Finally, by virtue of the Kalman-Yakubovich-Popov Lemma and the S-procedure, the frequency domain conditions are converted back to time-domain conditions, taking the form of tractable LMI-problems. The main concern is that the Plancherel Theorem does not exist in $L_1$ and we cannot expect to switch from time to frequency domain in the same spirit as in $L_2$. It is however still possible to consider the frequency domain as stated in the following result:
\begin{lemma}\label{lem:freq}
  The following equivalent statements hold:
  \begin{enumerate}
    \item[i)] The ILC
    \begin{equation}\label{eq:ILCprop}
          \int_0^{+\infty}\left(\varphi_{1}^Tw(t)+\varphi_{2}^Tz(t)\right)dt\ge0
    \end{equation}
    holds for every pairs of positive signals $(w,z)$ such that $z=\Sigma w$, where $\Sigma$ is a positive operator.
    \item[ii)] The algebraic inequality
    \begin{equation}\label{eq:algineq}
      \varphi_1^T\widehat{w}(0)+\varphi_2^T\widehat{z}(0)\ge0
    \end{equation}
    holds where $\widehat{w}$ and $\widehat{z}$ are the Laplace transform of the signals $w$ and $z$, respectively.
  \end{enumerate}
\end{lemma}
\begin{proof} The proof follows from noting that
  \begin{equation}
    \begin{array}{lcl}
      \int_0^{+\infty} \left(\varphi_{1}^Tw(t)+\varphi_{2}^Tz(t)\right)dt&=&\left.\int_0^{+\infty}\left[\left( \varphi_{1}^Tw(t)+\varphi_{2}^Tz(t)\right)e^{-st}\right]dt\right|_{s=0}\\
&=&\varphi_1^T\widehat{w}(0)+\varphi_2^T\widehat{z}(0).
    \end{array}
  \end{equation}
\end{proof}
Assuming that the pair of nonnegative signals $(w,z)$ is related by the equality $z=\Sigma_{LTI}^+ w$ where $\Sigma_{LTI}^+ $ is a \emph{linear time-invariant} positive operator, we have the following corollary:
\begin{corollary}\label{cor:freq}
  The following equivalent statements hold:
  \begin{enumerate}
    \item[i)] The ILC
    \begin{equation}\label{eq:ILCprop}
          \int_0^{+\infty}\left(\varphi_{1}^Tw(t)+\varphi_{2}^Tz(t)\right)dt\ge0
    \end{equation}
    holds for every pairs of nonnegative signals $(w,z)$, $z=\Sigma_{LTI}^+ w$, where $\Sigma_{LTI}^+$ is a linear time-invariant positive operator.
    \item[ii)] The algebraic inequality
    \begin{equation}\label{eq:LTIineq}
      \varphi_1^T+\varphi_2^T\widehat{\Sigma}_{LTI}^+(0)\ge0
    \end{equation}
    holds where $\widehat{\Sigma}_{LTI}^+(s)$ is the transfer function corresponding to the linear time-invariant positive operator $\Sigma_{LTI}^+$.
  \end{enumerate}
\end{corollary}
\begin{proof}  Since $\Sigma_{LTI}^+$ is a linear time-invariant positive operator, then the left-hand side of the inequality (\ref{eq:algineq}) can be written as $\left[\varphi_1^T+\varphi_2^T\widehat{\Sigma}(0)\right]\widehat{w}(0)$. The nonnegativity of the letter expression is equivalent to the nonnegativity of the term into brackets since $\widehat{w}(0)=\int_0^\infty w(t)dt\ge0$. This concludes the proof.
\end{proof}

\subsubsection{Generic robustness results and remarks}\ \\

\vspace{-3mm}\noindent Lemma \ref{lem:freq} and Corollary \ref{cor:freq} show that, similarly as for IQCs, the scalings may be selected both in the time-domain and the frequency domain. In the LTI case, the scalings just have to be selected according to the static-gain matrix of the system, which shows that the problem reduces to a problem with constant parametric uncertainties since the static-gain matrix is a (possibly uncertain) real matrix. Another conclusion is that, since only the zero-frequency is important, there is thus no need for capturing the entire frequency domain by using frequency-dependent scalings as in the $\mu$-analysis \cite{Zhou:02}, IQC-based techniques \cite{RantzerMegretski:97} or approaches using dynamic D-scalings \cite{Scherer:07b}. It is hence expected that a constant scaling matrix acting on the static-gain matrix should be enough for obtaining interesting results.

Along these lines, it is proved in \cite{Ebihara:11} that the use of the scaled $L_1$-gain yields necessary and sufficient conditions for the characterization of any interconnection of LTI positive systems. Similar results are obtained in the $L_2$ framework in \cite{Tanaka:11}. These results can be understood through the fact that the stability of interconnections of LTI positive systems is equivalent to the stability of interconnections of their static-gain matrix, i.e. higher order dynamics have no impact on stability. The critical stability information is hence concentrated in the static part of the dynamical systems.

Based on the above remarks, generic robustness properties are formalized below:
%
%
\begin{theorem}\label{th:robgen}
  The following statements are equivalent:
  \begin{enumerate}
    \item[i)] The linear time-invariant positive operator $\Sigma$ satisfies the ILC (\ref{eq:ILCprop}).
    \item[ii)] The static-gain matrix $\widehat{\Sigma}(0)$ satisfies (\ref{eq:LTIineq}).
    \item[iii)] The linear time-invariant positive operator $\Sigma^\prime$ satisfies the ILC (\ref{eq:ILCprop}), where $\widehat{\Sigma}^\prime(s)=\widehat{\Sigma}(s)+s\Theta(s)$ and $s\Theta(s)$ is any positive asymptotically stable proper transfer function.
  \end{enumerate}
\end{theorem}
\begin{proof}
The equivalence between i) and ii) follows from Corollary \ref{cor:freq}. The proof of ii) $\Rightarrow$ i) follows from choosing $\Theta(s)=0$. To show the converse, it is enough to remark that $\widehat{\Sigma}^\prime(0)=\widehat{\Sigma}(0)$ and hence $\Sigma^\prime$ satisfies the ILC (\ref{eq:ILCprop}) for any positive asymptotically stable strictly proper transfer function $\Theta(s)$.
\end{proof}
The above result interestingly shows that the set of uncertainties that satisfies the ILC is typically very large since there is basically no limit on the magnitude of the coefficients of $\Theta(s)$ acting on powers of $s$. 
%



\subsection{Main results}

Based on the results and discussions of the previous sections, the main results on robust stability analysis of linear positive systems subject to uncertain parametric uncertainties can be stated. Both $L_1$- and $L_\infty$-induced norms are considered. The conditions are expressed as robust linear optimization/feasibility problems that are difficult to solve directly. An exact solving scheme based on Handelman's Theorem is  proposed in Section \ref{sec:solvingunc}.

\begin{theorem}\label{th:robL1}
  Assume there exist a vector $\lambda\in\mathbb{R}^n_{++}$ and polynomials $\varphi_1(\delta),\varphi_2(\delta)\in\mathbb{R}^{n_0}$ such that the robust linear program
  \begin{equation}\label{eq:robstab1_L1}
  \begin{array}{rcl}
    \lambda^TA+\varphi_1(\delta)^TC_0+\mathds{1}_q^TC_1&<&0\\
    \lambda^TE_0+\varphi_2(\delta)^T+\varphi_1(\delta)^TF_{00}+\mathds{1}_q^TF_{10}&<&0\\
    \lambda^TE_1-\gamma\mathds{1}_p^T+\varphi_1(\delta)^TF_{01}+\mathds{1}_q^TF_{11}&<&0\\
  \end{array}
  \end{equation}
  \begin{equation}\label{eq:robstab2_L1}
    \varphi_1(\delta)^T+\varphi_2(\delta)^T\Delta(\delta)\ge0
  \end{equation}
  is feasible for all $\delta\in\deltab$. Then, the uncertain linear positive system (\ref{eq:unsyst}) is asymptotically stable and the $L_1$-gain of the transfer $w_1\to z_1$ is smaller than $\gamma>0$.
\end{theorem}

\begin{proof}
The proof follows exactly the same lines as for Lemma \ref{lem:L1}. Note however that the supply-rate
\begin{equation}
  s(w,z)=-\varphi_1(\delta)^Tz_0(t)-\varphi_2(\delta)^Tw_0(t)+\gamma \mathds{1}_p^Tw_1(t)-\mathds{1}_q^Tz_1(t)
\end{equation}
has to be considered here.
\end{proof}

It is important to mention that the above condition is unlikely a necessary condition since the linear Lyapunov function does not depend on the parameters. It is indeed well-known that quadratic stability (parameter independent Lyapunov function) is more conservative than robust stability (parameter dependent Lyapunov function).

The following result addresses the case when the ILC can be saturated using constant scalings:
\begin{theorem}\label{th:noncons}
The following statements are equivalent:
\begin{enumerate}
  \item[i)] System (\ref{eq:LFR}) with $\Delta=\Delta_0$, $\Delta_0\ge0$ constant, is asymptotically stable and the $L_1$-gain of the transfer $w\to z$ is smaller than $\gamma$.
  \item[ii)] System (\ref{eq:LFR}) with $\widehat{\Delta}(s)$ LTI, asymptotically stable, positive and $\widehat{\Delta}(0)=\Delta_0\ge0$ constant, is asymptotically stable and the $L_1$-gain of the transfer $w\to z$ is smaller than $\gamma$.
  \item[iii)] There exist vectors $\lambda\in\mathbb{R}^n_{++}$, $\varphi_1,\varphi_2\in\mathbb{R}^{n_0}$ such that the linear program
  \begin{equation}\label{eq:robstab1_L1b}
  \begin{array}{rcl}
    \lambda^TA+\varphi_1^TC_0+\mathds{1}_q^TC_1&<&0\\
    \lambda^TE_0+\varphi_2^T+\varphi_1^TF_{00}+\mathds{1}_q^TF_{10}&<&0\\
    \lambda^TE_1-\gamma\mathds{1}_p^T+\varphi_1^TF_{01}+\mathds{1}_q^TF_{11}&<&0\\
  \end{array}
  \end{equation}
  \begin{equation}\label{eq:robstab2_L1b}
    \varphi_1^T+\varphi_2^T\Delta_0=0
  \end{equation}
  is feasible.
\end{enumerate}
\end{theorem}
\begin{proof}
  The equivalence between i) and ii) follows from Theorem \ref{th:robgen}. It is easy to show from dissipativity theory that iii) implies i). This can also be proved by substituting $\varphi_1^T$ by $-\varphi_2^T\Delta_0$ in the inequalities (\ref{eq:robstab1_L1b}) and further substituting $\varphi_2^T$ by $-\lambda^TE_0(I-\Delta_0F_{00})^{-1}-\mathds{1}_q^TF_{10}(I-\Delta_0F_{00})^{-1}$. This leads to the inequalities
  \begin{equation}
  \begin{array}{lcl}
    \lambda^T(A+E_0(I-\Delta_0F_{00})^{-1}\Delta_0C_0)+\mathds{1}_q^T\left(C_1+F_{10}(I-\Delta_0F_{00})^{-1}\Delta_0C_0\right)&<&0\\
    \lambda^T(A+E_0(I-\Delta_0F_{00})^{-1}\Delta_0F_{01})-\gamma\mathds{1}_p^T+\mathds{1}_q^T\left(F_{11}+F_{10}(I-\Delta_0F_{00})^{-1}\Delta_0F_{01}\right)&<&0
  \end{array}
  \end{equation}
  which are exactly the conditions for asymptotic stability and bounded $L_1$-gain. By reversing the calculations, we can show that i) implies iii).
\end{proof}
This result can indeed be extended to the $L_\infty$-case, this is omitted for brevity.

\begin{theorem}\label{th:robLinf}
  Assume there exist a vector $\lambda\in\mathbb{R}^n_{++}$ and polynomials $\varphi_1(\delta),\varphi_2(\delta)\in\mathbb{R}^{n_0}$ such that the robust linear program
    \begin{equation}\label{eq:robstab1_Linf}
  \begin{array}{rcl}
    A\lambda+\bar{C}_0\varphi_1(\delta)+E_1\mathds{1}_p&<&0\\
    \tilde{E}_0^T\lambda+\varphi_2(\delta)+\bar{F}_{00}^T\varphi_1(\delta)+\tilde{F}_{10}^T\mathds{1}_p&<&0\\
    C_1\lambda-\gamma \mathds{1}_q+\bar{F}_{01}\varphi_1(\delta)+F_{11}\mathds{1}_p&<&0\\
  \end{array}
  \end{equation}
  \begin{equation}\label{eq:robstab2_Linf}
    \varphi_1(\delta)^T+\varphi_2(\delta)^T\Delta(\delta)\ge0,\ \delta\in\deltab
  \end{equation}
  is feasible for all $\delta\in\deltab$. Then, the uncertain linear positive system (\ref{eq:unsyst}) is asymptotically stable and the $L_\infty$-gain of the transfer $w_1\to z_1$ is smaller than $\gamma>0$.
\end{theorem}

\begin{proof}
  The proof is based on the use the LFT system corresponding of the transposed system of (\ref{eq:unsyst}) given by:
  \begin{equation}\label{eq:LFR2}
  \begin{array}{lcl}
    \dot{x}(t)&=&A^Tx(t)+\tilde{E}_0w_0(t)+C_1^Tw_1(t)\\
    z_0(t)&=&\bar{C}_0x(t)+\bar{F}_{00}w_0(t)+\bar{F}_{01}w_1(t)\\
    z_1(t)&=&E_1^Tx(t)+\tilde{F}_{10}w_0(t)+F_{11}^Tw_1(t)\\
    w_0(t)&=&\Delta(\delta)z_0(t)
  \end{array}
\end{equation}
where the matrices $\tilde{F}_{10}$ and $\tilde{E}_0$ are specific matrices of the transposed system. All the other matrices are those of systems (\ref{eq:unsyst}) and (\ref{eq:LFR}).
\end{proof}

\begin{remark}
It must be stressed here that the Linear Fractional Transformation does not commute with the operation of transposition. In other words, the transposed of an LFT system does not coincide, in general, with the LFT of the transposed system. Some matrices may indeed remain unchanged (non transposed) while some others are different. This has motivated the use of the 'bar' and 'tilde' notations in (\ref{eq:LFR2}). For instance, when the system depends polynomially on the parameters, we may have the equalities $\bar{F}_{00}=F_{00},\bar{F}_{01}=F_{01}$ and $\bar{C}_0=C_0$.
\end{remark}

\section{Robust Stabilization}\label{sec:stabip}

The robust stabilization problem with $L_\infty$-performance is solved in this section. It is interesting to note that the presence of scalings never destroys the convexity of the approach as this may occur in some robust control approaches based on IQC's \cite{RantzerMegretski:97} or the full-block S-procedure \cite{Scherer:97}.

As in Section \ref{sec:stabiunp}, it is not necessary that the open-loop be positive, the only requirement is the nonnegativity of disturbance input matrices, i.e. $E(\delta),F(\delta)$ nonnegative for all $\delta\in\deltab$. The LFR of the closed-loop transposed system is given by:
\begin{equation}\label{eq:LFRcl}
  \begin{array}{lcl}
    \dot{\tilde{x}}(t)&=&\mathcal{A}(K)\tilde{x}(t)+\mathcal{E}_0(K)\tilde{w}_0(t)+\mathcal{E}_1(K)\tilde{w}_1(t)\\
    \tilde{z}_0(t)&=&\mathcal{C}_0\tilde{x}(t)+\mathcal{F}_{00}\tilde{w}_0(t)+\mathcal{F}_{01}\tilde{w}_1(t)\\
    \tilde{z}_1(t)&=&\mathcal{C}_1\tilde{x}(t)+\mathcal{F}_{10}\tilde{w}_0(t)+\mathcal{F}_{11}\tilde{w}_1(t)\\
    \tilde{w}_0(t)&=&\Delta(\delta)^T\tilde{z}_0(t)
  \end{array}
\end{equation}
where $\mathcal{A}(K)=(A^0+B^0K)^T$, $ \mathcal{E}_1(K)=(C_1^0+D^0K)^T$ and
\begin{equation}
     \mathcal{E}_0(K)=\begin{bmatrix}
       (A^1+B^1K)^T & \ldots & (A^\eta+B^\eta K)^T & (C_1^1+D^1K)^T & \ldots & (C_1^\theta+D^\theta K)^T
     \end{bmatrix}.
\end{equation}
The integers $\eta,\theta>0$ are related to the system dependency on the parameters.

\begin{theorem}\label{th:stabilizationz}
 The controlled system (\ref{eq:unsyst})-(\ref{eq:cl}) with $K\in\mathcal{K}$ is asymptotically stable if there exist vectors $\lambda\in\mathbb{R}^n_{++}$, $\varphi_1(\delta),\varphi_2(\delta)\in\mathbb{R}^{n_0}$,  $\mu_i\in\mathbb{R}^m$, $i=1,\ldots,n$ and a scalar $\gamma>0$ such that the robust linear program
  \begin{equation}\label{eq:stabzlol}
    \begin{array}{rcl}
    A^0\lambda+B^0\mu+\mathcal{C}_0^T\varphi_1(\delta)+\mathcal{C}_1^T\mathds{1}_p&<&0\\
    \tilde{\mathcal{E}}_0(\lambda,\mu)+\varphi_2(\delta)+\mathcal{F}_{00}^T\varphi_1(\delta)+\mathcal{F}_{10}^T\mathds{1}_p&<&0\\
    C_1^0\lambda+D^0\mu-\gamma\mathds{1}_q+\mathcal{F}_{01}^T\varphi_1(\delta)+\mathcal{F}_{11}\mathds{1}_p&<&0\\
    \varphi_1^T+\varphi_2^T\Delta(\delta)&\ge&0
  \end{array}
  \end{equation}
  \begin{equation}\label{eq:semiinfcst}
    \begin{array}{lcl}
      [A_\delta(\delta)]_{ij}\lambda_j+[B_\delta(\delta)]_{r,i}\mu_j&\ge&0,\ i,j=1,\ldots, n,\ i\ne j\\
      \left[C_\delta(\delta)\right]_{ij}\lambda_j+[D_\delta(\delta)]_{r,i}\mu_j&\ge&0,\ i=1,\ldots,q,\ j=1,\ldots, n
    \end{array}
  \end{equation}
is feasible for all $\delta\in\deltab$ with $\mu=\sum_{i=1}^n\mu_i$, $\tilde{\mathcal{E}}_0=\begin{bmatrix}
  \tilde{\mathcal{E}}_0^1 & \tilde{\mathcal{E}}_0^2
\end{bmatrix}$ and
  \begin{equation}
  \begin{array}{lcl}
    \tilde{\mathcal{E}}_0^1(\lambda,\mu)&=&\begin{bmatrix}
      A^1\lambda+B^1K\mu & \ldots & A^\eta\lambda+B^\eta K\mu
    \end{bmatrix},\\
    \tilde{\mathcal{E}}_0^2(\lambda,\mu)&=&\begin{bmatrix}
      C_1^1\lambda+D^1K\mu & \ldots & C_1^\theta\lambda+D^\theta K\mu
    \end{bmatrix}.
  \end{array}
  \end{equation}
  Moreover, in such a case, the controller $K$ is given by
       \begin{equation}
         K=\begin{bmatrix}
           \lambda_1^{-1}\mu_1 & \ldots & \lambda_n^{-1}\mu_n
         \end{bmatrix}
       \end{equation} and the closed-loop system satisfies $||z_1||_{ L_\infty}\le\gamma ||w_1||_{ L_\infty}$.
\end{theorem}
\begin{proof}
  The proof is similar to the one of the results in Section \ref{sec:stabiunp}.
\end{proof}

The constraints (\ref{eq:stabzlol}) are polynomial in the uncertain parameters $\delta$. This is however not the case of the constraints (\ref{eq:semiinfcst}) that are rational when the system is rational. They can however be easily turned into polynomial constraints by finding a common denominator for the left-hand side. Since the sign of the common denominator is fixed (otherwise the system would be ill-posed), the rational constraints reduce to polynomial constraints on the numerators.

\section{Solving robust linear programs}\label{sec:solvingunc}

In this section, we address the problem of solving the robust linear optimization problems arising in Theorems \ref{th:robL1}, \ref{th:robLinf} and \ref{th:stabilizationz}. A solving scheme based on Handelman's Theorem \cite{Handelman:88} is proposed.

\subsection{Handelman's Theorem}

For completeness, let us first recall Handelman's Theorem:
\begin{theorem}[Handelman's Theorem]
  Assume $\mathcal{S}$ is a compact polytope in the Euclidean $N$-space defined as ${\mathcal{S}:=\left\{x\in\mathbb{R}^N:\ g_i(x)\ge0,\ i=1,\ldots\right\}}$ where the $g_i(x)$'s are linear forms. Assume also that $P$ is a polynomial in $N$ variables which is positive on $\mathcal{S}$, then $P$ can be expressed as a linear combination with nonnegative coefficients (not all zero) of products of members of $\{g_i\}$.
\end{theorem}

The above theorem states a necessity result regarding the positivity of a polynomial over a compact polytope: if it is positive, then we can write it as a linear combination of products of the linear functions defining $\mathcal{S}$. In \cite{Powers:01,Scheiderer:09} it is shown that Handelman's Theorem implies P\'{o}lya's Theorem \cite{Hardy:78}. Sufficiency is immediate following similar arguments as for the S-procedure \cite{Boyd:94a} or sum-of-squares techniques \cite{Parrilo:00}. A linear combination of positive polynomials on $\mathcal{S}$ (the products of $g_i$'s are positive on $\mathcal{S}$) is indeed a positive polynomial on $\mathcal{S}$. Hence, with a suitable choice of the product terms, it is possible to determine whether a polynomial is positive over a compact polytope of the Euclidian space. To illustrate the above statement, let us consider the following example:
\begin{example}\label{ex:handel}
Suppose we would like to characterize all the univariate polynomials $p(x)$ of at most degree 2 that are nonnegative on the interval $[-1,1]$. The basis functions of this interval are given by $g_1(x)=x+1$ and $g_2(x)=1-x$. According to Handelman's Theorem, we know that all such polynomials write as a linear combination of all possible products $g_1(x)^ig_2(x)^j$ with $1\le i+j\le2$. Hence, we have
  \begin{equation}
  \begin{array}{lcl}
        p(x)&=&\tau_1g_1(x)+\tau_2g_2(x)+\tau_3g_1(x)g_2(x)+\tau_4g_1(x)^2+\tau_5g_2(x)^2\\
        &=&\chi_2x^2+\chi_1x+\chi_0
  \end{array}
  \end{equation}
  where $\tau_i\ge0$, $i=1,\ldots,5$ and
  \begin{equation}\label{eq:mdrlol2}
    \begin{array}{lcl}
      \chi_2&=&\tau_4+\tau_5-\tau_3,\\
      \chi_1&=&\tau_1-\tau_2+2\tau_4-2\tau_5,\\
      \chi_0&=&\tau_1+\tau_2+\tau_3+\tau_4+\tau_5.
    \end{array}
  \end{equation}
  Hence, we can conclude on the very general statement that any univariate polynomial of at most degree 2 that is nonnegative on a bounded interval of the form $[\alpha,\beta]$ can be viewed as a point $\tau=(\tau_1,\ldots,\tau_5)\in\mathbb{R}_+^5$.
\end{example}


\subsection{Equivalent relaxations for robust linear programs}

For simplicity, we will focus here on linear optimization problems depending on a single uncertain parameter. The results straightforwardly generalize to the case of multiple uncertain parameters at the expense of complex notations. Let us consider the following semi-infinite feasibility problem
\begin{problem}\label{pb:pb}
  There exists $x\in\mathbb{R}^\eta$ such that the inequality
  \begin{equation}
    P(x,\theta):=\sum_{j=0}^{d}P_j(x)\theta^{j}\le0
  \end{equation}
  holds for all ${\theta\in\mathcal{P}:=\mathcal{P}=\left\{\theta\in\mathbb{R}:g_i(\theta)\ge0,\ i=1,\ldots,\ g_i's\ \mathrm{linear}\right\}}$ and where the vectors $P_j(x)\in\mathbb{R}^{N_P}$, $j=0,\ldots,d$, are linear in $x$.
\end{problem}
We then have the following result:
\begin{theorem}
  The following statements are equivalent:
  \begin{enumerate}
    \item[i)] Problem \ref{pb:pb} is feasible.
    \item[ii)] There exist an integer $b>0$ and vectors\footnote{The notation $y$ is here to emphasize that the vectors $Q_k$'s consist exclusively of independent additional decision variables.} $Q_k(y)\in\mathbb{R}^{N_P}$, $k=1,\ldots,b$, such that the finite-dimensional linear program
          \begin{subequations}\label{eq:relax1}
            \begin{eqnarray}
                P_j(x)&=&Z_j,\ j=0,\ldots,d\label{eq:relax1_1}\\
                Q_k(y)&\le&0,\ k=0,\ldots,b\label{eq:relax1_2}
            \end{eqnarray}
        \end{subequations}
        is feasible in $(x,y)\in\mathbb{R}^\eta\times\mathbb{R}_+^{N_p(b+1)}$ where $P(x,\theta) = \sum_{k=0}^{b}Q_k(y)g_i(\theta)^k$, ${Z_j:=\sum_{k=0}^{b}\upsilon_{jk}Q_k(y)}$ and the $\upsilon_{jk}$'s depend on the coefficients of the basis functions $g_i$'s.
    \item[iii)] There exist an integer $b>0$ and vectors $R_k(z)\in\mathbb{R}^{N_P}$, $k=0,\ldots,b-d$, such that the finite-dimensional linear program
        \begin{subequations}\label{eq:relax2}
            \begin{align}
            R_k(z)&\le0,\ k=0,\ldots,d\label{eq:relax2_1}\\
             \Upsilon_2^{-1}\left(\begin{bmatrix}
              P_0(x)\\
              \vdots\\
              P_{d}(x)
            \end{bmatrix}-\Upsilon_1\begin{bmatrix}
              R_0(z)\\
              \vdots\\
              R_{b-d}(z)
            \end{bmatrix}\right)&\le0.\label{eq:relax1_2}
            \end{align}
        \end{subequations}
        is feasible in $(x,z)\in\mathbb{R}^\eta\times\mathbb{R}_+^{N_p(b-d+1)}$ where $\Upsilon:=\begin{bmatrix}
          \Upsilon_1 & \Upsilon_2
        \end{bmatrix}=[\upsilon_{ij}]$, $\Upsilon_1\in\mathbb{R}^{(d+1)\times(b-d)}$, and $\Upsilon_2\in\mathbb{R}^{(d+1)\times (d+1)}$.
  \end{enumerate}
\end{theorem}
\begin{proof}
  The equivalence between statements i) and ii) is an immediate consequence of Handelman's Theorem. The equivalence between statements ii) and iii) follows from simple algebraic manipulations allowing to reduce the number of additional decision variables. First remark that the equality constraints can be compactly written as
\begin{equation*}
  \begin{bmatrix}
              P_0(x)\\
              \vdots\\
              P_{d}(x)
  \end{bmatrix}=\Upsilon\begin{bmatrix}
              Q_0(y)\\
              \vdots\\
              Q_{b}(y)
  \end{bmatrix}.
\end{equation*}
Note that the matrix $\Upsilon$ is full-row rank, otherwise it would not be possible to characterize independent polynomial coefficients. Using the decomposition $\Upsilon=\begin{bmatrix}
  \Upsilon_1 & \Upsilon_2
\end{bmatrix}$ where $\Upsilon_2\in\mathbb{R}^{(d+1)\times (d+1)}$ is w.l.o.g. a nonsingular matrix, the equality constraints can then be solved to get
\begin{equation}
\begin{bmatrix}
Q_{1}(y)\\
\vdots\\
Q_{d}(y)
\end{bmatrix}=\Upsilon_2^{-1}\left(\begin{bmatrix}
              P_0(x)\\
              \vdots\\
              P_{d}(x)
            \end{bmatrix}-\Upsilon_1\begin{bmatrix}
              Q_{d+1}(y)\\
              \vdots\\
              Q_{b}(y)
            \end{bmatrix}\right).
\end{equation}
Since the $Q_i(y)$'s are nonnegative vectors, this is then equivalent to say that the right-hand side of the above equality is nonnegative. Finally, posing $R_k=Q_{d+k}$, $k=1,\ldots b-d$, we obtain the feasibility problem of statement iii). The opposite implication is obtained by reverting the reasoning. The proof is complete.
\end{proof}

In the light of the above result, it turns out that it is always possible to represent a robust linear program with polynomial dependence on uncertain parameters within a compact polytope as a more complex finite dimensional linear program. Handelman's Theorem yields a linear program involving $N_P(b+1)+\eta$ variables, $N_P(d+1)$ equality constraints and $N_P(b+1)$ inequality constraints. After reduction, the problem has $N_P(b-d)+\eta$ decision variables and $N_P(b+2)$ inequality constraints. The complexity of the problem has then be reduced. The same reasoning applies to the case of multiple uncertainties.

The question of choosing which products of basis functions to consider is a difficult problem. A brute force approach would consider all the possible products up to a certain degree. This is an easy task in the case of univariate polynomials. This is however more problematic when considering multivariate polynomials since the number of basis functions $b$ can be very large. A bound on the necessary order for the products has been provided in \cite{Powers:01} and generalizes straightforwardly to vector polynomials.

\section{Examples}\label{sec:ex}

\subsection{ILCs for several classes of uncertainties}


ILCs for some common operators are discussed below.\vspace{3mm}

\subsubsection{Uncertain SISO tranfer function}\ \\

\vspace{-3mm}\noindent Consider now an uncertain asymptotically stable positive SISO proper transfer function $\widehat{\Sigma}(s,\rho)$ depending on constant uncertain parameters $\rho\in[0,1]^N$. By virtue of Corollary \ref{cor:freq}, the ILC can be expressed as  $\phi_1^T+\phi_2^TZ\ge0$ where $Z\in\left\{\widehat{\Sigma}(0,\rho):\ \rho\in[0,1]^N\right\}$.\vspace{3mm}

\subsubsection{Multiplication operator}\label{ex:3e}\ \\

\vspace{-3mm}\noindent Let us consider in this example the multiplication operator $\Sigma$ which multiplies the input signal by a bounded and time-varying parameter $\delta(t)$ varying arbitrarily within its range of values, i.e. $z(t)=\Sigma(w)(t)=\delta(t)w(t)$. Since the parameter is time-varying, the time-domain version of the ILC must be considered
  \begin{equation}
    \int_0^{+\infty}[\varphi_1^T+\varphi_2^T\delta(\theta)]w(\theta)d\theta\ge0.
  \end{equation}
A necessary and sufficient condition is given by $\varphi_1^T+\varphi_2^T\delta(t)\ge0$ for all $t\ge0$ where the scalings $\varphi_1$ and $\varphi_2$ are chosen according to the range of values of $\delta(t)$. Note that it is also possible to select polynomial scalings verifying $\varphi_1(\delta)=-\delta\varphi_2(\delta)$ in order to saturate the ILC and obtain a better characterization of the uncertainty set.\vspace{3mm}

\subsubsection{Uncertain infinite dimensional system}\ \\

\vspace{-3mm}\noindent The same reasoning  also applies to general asymptotically-stable positive LTI infinite-dimensional systems governed by partial differential equations. Let us consider the heat equation given by
\begin{equation*}
\begin{array}{rcl}
    \dfrac{\partial u}{\partial t}&=&\omega^2\dfrac{\partial^2 u}{\partial x^2}\\
    u(0,t)&=&w(t)\\
    z(t)&=&u(1,t)
\end{array}
\end{equation*}
where $u(x,t)$ is the state of the system, $x\in[0,1]$ the space variable,  $w(t)$ the input, $z(t)$ the output and $\omega>0$ an uncertain parameter of the system. The transfer function of the system is given by $G(s)=\alpha e^{-\sqrt{s}/\omega}+\beta e^{\sqrt{s}/\omega}$ where $\alpha$ and $\beta$ are real constants determined according to initial and boundary conditions. The static-gain $G(0)=\alpha+\beta$ is independent of $\omega$ and, according to Corollary \ref{cor:freq}, it is enough to choose $\varphi_1=-\varphi_2(\alpha+\beta)$. Stability of interconnections of LTI positive finite-dimensional systems and the heat-equation can hence be proved regardless of $\omega\in[0,+\infty)$.\vspace{3mm}

\subsubsection{Constant delay operator}\ \\

\vspace{-3mm}\noindent Let us consider the constant delay operator $\widehat{\Sigma}(s)=e^{-sh}$, $h\ge0$, which is indeed a positive operator. The static-gain $\widehat{\Sigma}(0)$ is equal to 1 regardless of the value of $h$. The ILC can hence be easily saturated by choosing  $\varphi_1^T=-\varphi_2^T$. It is shown in Section \ref{ex:delaystab} that the consideration of such scalings leads to a necessary and sufficient condition for the stability of positive time-delay systems with constant delays, recovering then the results of \cite{Haddad:04}.\vspace{3mm}

\subsubsection{Time-varying delay operator}\label{ex:4e}\ \\

\vspace{-3mm}\noindent Let us consider now the time-varying delay operator $\Sigma$ defined as $z(t)=\Sigma(w)(t)=w(t-h(t))$ where $h(t)\ge0$ and $\dot{h}(t)\le\mu<1$ for all $t\ge0$. In such a case, the operator can only be characterized through its $L_1$-gain given by $(1-\mu)^{-1}$ under the standard assumption of 0 initial conditions \cite{GuKC:03}. A suitable ILC for this operator is then given by $$\int_0^\infty\varphi^T\left(w(t)-(1-\mu)z(t)\right)dt\ge0$$
with $\varphi>0$. Hence $\varphi_1=\varphi$ and $\varphi_2=-(1-\mu)\varphi$.\vspace{3mm}

\subsection{Computation of Norms}
%

For this example, many linear positive systems have been randomly generated and their induced-norms computed on a laptop equipped with an Intel U7300 processor of 1.3GHz with 4GB of RAM. The mean computation time and the standard deviation for different systems are gathered in Table \ref{tab:comptime}. Note that the number of variables is $n+1$ and the number of constraints is $2n+p+1$ and $2n+1+q$ for the $ L_1$-gain and the $ L_\infty$-gain respectively. Since the number of constraints is larger for the $L_\infty$-gain, it is expected that it takes longer to compute. It seems important to note that the the induced-norms are different since the system is not SISO, so the theoretical analysis in \cite{Rantzer:11} does not hold here.

  \begin{table*}
  \centering
  \begin{tabular}{c|c||c|c}
    nb. of systems & $(n,p,q)$ & $ L_1$-gain & $ L_\infty$-gain\\
    \hline
    \hline
    $20$ & (300,100,150) & $\mu=12.282$, $\sigma=1.1406$ & $\mu=14.186$, $\sigma=1.4151$\\
    $100$ & $(50,20,30)$ & $\mu=0.53973$, $\sigma=0.27486$ & $\mu=0.50735$, $\sigma=0.080446$
  \end{tabular}
  \caption{Mean computation time $\mu$ [sec] and standard deviation $\sigma$ [sec] for gain computation}\label{tab:comptime}
\end{table*}

\subsection{Example 2: Drug distribution}

A frequently used model to analyze the distribution or flow of a drug
or a tracer through the human body after injection into
the bloodstream is given by the following compartmental model \cite{Haddad:10}:
\begin{equation}\label{ex:drug}
  \dot{x}(t)=\begin{bmatrix}
    -(a_{11}+a_{21}) & a_{12}\\
    a_{21} & -a_{12}
  \end{bmatrix}x(t)+\begin{bmatrix}
    1\\
    0
  \end{bmatrix}u(t)
\end{equation}
where $x_1$, $x_2$ and $u$ are the compartment corresponding to the blood plasma, the extravascular space (e.g. tissue) and the drug/tracer injection. Here, the drug or the tracer is injected directly into the bloodstream and the drug is evacuated by the kidneys at rate $a_{11}>0$. The scalars $a_{12},a_{21}>0$ are the transmission coefficients for the drug/tracer between the bloodstream and the extravascular space. Due to the sign pattern, the system is positive and asymptotically stable. When the system is SISO the $L_1$-gain coincides with the $L_\infty$-gain and they are both equal to the static-gain of the system. Different gains for different output matrices are given in Table \ref{tab:drug}. 

\begin{table}
\centering
  \begin{tabular}{c||c|c}
    Output matrix $C$ & $L_1$-gain & $L_\infty$-gain\\
    \hline
    \hline
    \vspace{-3mm} \\
    $C=\begin{bmatrix}
      1 & 0
    \end{bmatrix}$ & $\frac{1}{a_{11}}$ & $\frac{1}{a_{11}}$\\
    \vspace{-3mm} \\
        $C=\begin{bmatrix}
      0 & 1
    \end{bmatrix}$ & $\frac{a_{21}}{a_{11}a_{12}}$ & $\frac{a_{21}}{a_{11}a_{12}}$\\
    $C=\diag(k_1,k_2)$ & $\frac{|k_1|}{a_{11}}+\frac{|k_2|a_{21}}{a_{11}a_{12}}$ & $\max\left\{\frac{|k_1|}{a_{11}},\frac{|k_2|a_{21}}{a_{11}a_{12}}\right\}$
  \end{tabular}
  \caption{Gains of system (\ref{ex:drug}) for different output matrices}\label{tab:drug}
\end{table}

\subsection{Theoretical Robustness analysis - Time-delay systems}\label{ex:delaystab}

We illustrate here the nonconservativeness of Theorem \ref{th:noncons} made possible by the saturation of the ILC condition. The case of time-delay systems is addressed.\vspace{3mm}

\subsubsection{Constant time-delay}\ \\

\vspace{-3mm}\noindent Let us consider the linear positive system with constant time-delay:
\begin{equation}
  \dot{x}(t)=Ax(t)+A_hx(t-h)
\end{equation}
for some $h\ge0$. It is well known that such a system is positive if and only if the matrix $A$ is Metzler and the matrix $A_h$ is nonnegative. Rewriting it in an LFT form we get
\begin{equation}
  \begin{array}{lcl}
    \dot{x}(t)&=&Ax(t)+A_hw_0(t)\\
    z_0(t)&=&x(t)\\
    w_0(t)&=&\nabla_c(z_0)(t)
  \end{array}
\end{equation}
where $\nabla_c$ is the constant delay operator with Laplace transform $\widehat{\nabla}_c(s)=e^{-sh}$. Since the static-gain of the delay operator is equal to 1, we can apply Theorem \ref{th:noncons} which yields the stability conditions
\begin{equation}
  \begin{array}{rcl}
    \lambda^TA+\varphi_1^T&<&0,\\
    \lambda^TA_h+\varphi_2^T&<&0,\\
    \varphi_1^T+\varphi_2^T&=&0
  \end{array}
\end{equation}
that must be feasible for some $\lambda\in\mathbb{R}_{++}^n$ and $\varphi_1,\varphi_2\in\mathbb{R}^n$. They are equivalent to the conditions
\begin{equation}\label{eq:split}
  \lambda^TA+\varphi_1^T<0\ \mathrm{and}\ \lambda^TA_h-\varphi_1^T<0
\end{equation}
which are in turn equivalent\footnote{To see the equivalence, note that (\ref{eq:split}) implies (\ref{eq:fusion}) by summation. To prove the converse, assume (\ref{eq:fusion}) holds and defining $\varphi_1^T=\lambda^TA_h-\lambda^T(A+A_h)/2>0$ in (\ref{eq:split}) makes the conditions negative.} to the inequality
\begin{equation}\label{eq:fusion}
  \lambda^T(A+A_h)<0
\end{equation}
well-known to be a necessary and sufficient condition for asymptotic stability of positive time-delay systems \cite{Haddad:04, Kaczorek:09}. Note also that this condition is dual to the one derived in \cite{AitRami:09}. By solving for $\lambda$ instead of $\varphi_1$, we get the condition
\begin{equation}\label{eq:temp_2}
  -\varphi_1^TA^{-1}A_h<\varphi_1^T,\ \varphi_1>0
\end{equation}
which is a linear counterpart of the spectral radius condition for the stability of linear time-delay systems, see e.g. \cite{GuKC:03}.\vspace{3mm}

\subsubsection{Time-varying delay}\ \\

\vspace{-3mm}\noindent Consider now the time-varying delay operator which has been considered in the case of positive systems in \cite{AitRami:09}. The $L_1$-gain of time-varying operator is $(1-\mu)^{-1}$, $\mu<1$ as discussed in Section \ref{ex:4e} . In such a case, the stability conditions are given by
\begin{equation}\label{eq:temp_1}
  \begin{array}{rcl}
    \lambda^TA+\varphi^T&<&0,\\
    \lambda^TA_h-(1-\mu)\varphi^T&<&0
  \end{array}
\end{equation}
for some $\lambda,\varphi\in\mathbb{R}_{++}^n$.
Identical conditions can be obtained using a linear Lyapunov-Krasovskii functional \cite{Kaczorek:09} of the form $ V(x_t)=\lambda^Tx(t)+\varphi^T\int_{t-h(t)}^tx(s)ds$
with $\lambda,\varphi>0$. Note that in the robust formulation, (\ref{eq:temp_1}) is feasible only if $\varphi>0$. Note also that the conditions can be merged into the novel single inequality
\begin{equation}
  \lambda^T((1-\mu)A+A_h)<0
\end{equation}
where we can see that the time-varying delay penalizes the stability by scaling down the matrix $A$. The time-invariant case is retrieved for $\mu=0$.

While the above results are based on the $L_1$-gain, it seems interesting to analyze the stability using the $L_\infty$-gain. The $L_\infty$-gain of the time-varying delay operator is equal to 1 \cite{Briat:11j}. We may then use a small-gain argument to prove asymptotic stability of the time-delay system by applying a scaled version of Lemma \ref{lem:Linf}. Hence, the time-delay system is asymptotically stable provided that there exist $\lambda,\varphi\in\mathbb{R}_{++}^n$ such that the conditions
\begin{equation}
  \begin{array}{lclclcl}
    A\lambda+A_h\varphi&<&0&\mathrm{and}&\lambda-\varphi^T&<&0
  \end{array}
\end{equation}
hold. They are equivalent to the condition $(A+A_h)\lambda<0$ (stability for zero delay) or the condition $-A^{-1}A_h\varphi<\varphi$ (small-gain condition) which is the $L_\infty$-gain counterpart of (\ref{eq:temp_2}). It is important to precise that the stability condition is identical to the one for systems with constant delay. This result is then much stronger than the $L_1$-based one since the delay-derivative bound does not have any negative impact on the stability. As discussed in \cite{Briat:11j}, while results using the $L_1$-norm relate to Lyapunov-Krasovskii functionals, those based on the $L_\infty$-norm connect to Lyapunov-Razumikhin functions. Note however that no Lyapunov-Razumikhin results for positive systems have been reported so far.


\subsection{Numerical robustness analysis - A systems biology example}

Let us consider the example of a simple gene expression process described by the following model \cite{Khammash:10}:
\begin{equation}
  \begin{bmatrix}
    \dot{x}_r(t)\\
    \dot{x}_p(t)\\
  \end{bmatrix}=\begin{bmatrix}
    -\gamma_r & 0\\
    k_p & -\gamma_p
  \end{bmatrix}\begin{bmatrix}
    x_r(t)\\
    x_p(t)\\
  \end{bmatrix}+\begin{bmatrix}
    1\\
    0
  \end{bmatrix}u(t)
\end{equation}
where $x_r\ge0$ is the mean number of mRNA in the cell, $x_p\ge0$ is the mean number of protein of interest in the cell. Above $u(t)\ge0$ is the transcription rate of DNA into mRNA, $\gamma_r>0$ is the degradation rate of mRNA, $k_p>0$ is the translation rate of mRNA into protein and $\gamma_p>0$ is the degradation rate of the protein. The parameters are assumed to be uncertain and given by $\gamma_r=\gamma_r^0+\eps_1\gamma_r^1$, $k_p=k_p^0+\eps_2k_p^1$ and $\gamma_p=\gamma_p^0+\eps_3\gamma_p^1$, where $\eps_i\in[-1,1]$, $i=1,2,3$. We are interested in analyzing the $L_\infty$-gain of the transfer from $u\to x_p$ over the set of all possible systems. In other terms, we would like to analyze the impact of the maximal value of the transcription rate to the maximal value of the mean number of proteins.

The overall system can be rewritten as
\begin{equation}\label{eq:uncbio}
  \begin{bmatrix}
    \dot{x}_r(t)\\
    \dot{x}_p(t)\\
  \end{bmatrix}=A_u(\eps)
  \begin{bmatrix}
    x_r(t)\\
    x_p(t)\\
  \end{bmatrix}+\begin{bmatrix}
    1\\
    0
  \end{bmatrix}u(t)
\end{equation}
where $\eps=\col(\eps_1,\eps_2,\eps_3)$ and
$$A_u(\eps)=\begin{bmatrix}
    -\gamma_r^0 & 0\\
    k_p^0 & -\gamma_p^0
  \end{bmatrix}+\eps_1\begin{bmatrix}
    -\gamma_r^1 & 0\\
    0 & 0
  \end{bmatrix}+\eps_2\begin{bmatrix}
    0 & 0\\
    k_p^1 & 0
  \end{bmatrix}+\eps_3\begin{bmatrix}
    0 & 0\\
    0 & -\gamma_p^1
  \end{bmatrix}.$$

Note that the theoretical computation of the $L_\infty$-gain is rather difficult since the static-gain depends rationally and in a nonconvex way on the uncertain parameter vector $\eps$. However, by exploiting the fact that the set $\mathscr{A}:=\left\{A_u(\eps):\eps\in[-1,1]^3\right\}$ is convex, the conditions  (\ref{eq:L1stability}), or equivalently (\ref{eq:Linfstability}), can be checked without using the LFT formulation\footnote{Remember that the interest of the LFT is to convexify the problem.}, although in this case both methods are equivalent provided that the scalings are chosen accordingly. Checking indeed the conditions (\ref{eq:L1stability}) over the whole set $\mathscr{A}$ is equivalent to checking them over the set of its vertices, i.e. $\mathscr{A}_v:=\left\{A_u(\eps):\eps\in\{-1,1\}^3\right\}$. This follows from standard convexity argument.

Let us consider $\gamma_r^0=1$, $k_p^0=2$ and $\gamma_p^0=1$ for numerical application. We also assume that the parameters are known up to a percentage $N\in[0,1)$ of their nominal value, hence $\gamma_r^1=N\gamma_r^0$, $\gamma_p^1=N\gamma_p^0$ and $k_p^1=Nk_p^0$. The results are summarized in Table \ref{tab:bio} where we can see that the proposed approach is rather accurate.

\begin{table}
\centering
  \begin{tabular}{c|c|c}
    $N$ & $L_\infty$-gain & Theoretical\\
    \hline
    \hline
    0 & 2 & 2\\
    0.1 & 2.7162 & 2.7161\\
    0.3 & 5.3063 & 5.3062\\
    0.5 & 12.0003 & 12.0000\\
    0.7 & 37.7783 & 37.7779
  \end{tabular}
  \caption{Computed and theoretical $L_\infty$-gains of the transfer $u\to x_p$ of the uncertain system (\ref{eq:uncbio})}\label{tab:bio}
\end{table}

\subsection{Numerical robustness analysis - A polynomial system example}\label{sec:numrobanal}

Let us consider the uncertain system with constant parametric uncertainty $\delta\in[0,1]$:
\begin{equation}
  \begin{array}{lcl}
    \dot{x}(t)&=&(A^0+\delta A^1+\delta^2A^2)x(t)+(E^0+\delta E^1+\delta^2E^2)w_1(t)\\
    z_1(t)&=&(C^0+\delta C^1+\delta^2C^2)x(t)+(F^0+\delta F^1+\delta^2F^2)w_1(t)
  \end{array}
\end{equation}
with the matrices
\begin{equation}
  \begin{array}{lclclclclcl}
    A^0&=&\begin{bmatrix}
      -10 & 2 & 4\\
      3 & -8 & 1\\
      2 & 1 & -5
    \end{bmatrix},&&A^1&=&\begin{bmatrix}
      1 & 0 & 2\\
      0 & 1 & 2\\
      -1 & 2 & -1
    \end{bmatrix},&&A^2&=&\begin{bmatrix}
      1 & -1 & -1\\
      1 & -1 & 0\\
      0 & 1 & -1
    \end{bmatrix},\\
    E^0&=&\begin{bmatrix}
      1 & 3\\
      3 & 0\\
      2 & 1
    \end{bmatrix},&&E^1&=&\begin{bmatrix}
      1 & 3\\
      1 & 1\\
      2 & 1
    \end{bmatrix},&&E^2&=&\begin{bmatrix}
      1 & 3\\
      0 & 1\\
      1 & 4
    \end{bmatrix},\\
    C^0&=&\begin{bmatrix}
      1 & 3 & 1\\
      2 & 0 & 1
    \end{bmatrix},&&C^1&=&\begin{bmatrix}
      1 & 0 & 2\\
      3 & 1 & 0
    \end{bmatrix},&&C^2&=&\begin{bmatrix}
      0 & 3 & 2\\
      1 & 4 & 1
    \end{bmatrix},\\
        F^0&=&\begin{bmatrix}
      2 & 1\\
      1 & 2
    \end{bmatrix},&&F^1&=&\begin{bmatrix}
      0 & 2\\
      1 & 0
    \end{bmatrix},&&F^2&=&\begin{bmatrix}
      1 & 1\\
      2 & 1
    \end{bmatrix}.
  \end{array}
\end{equation}
The system can be rewritten in the LFT form (\ref{eq:LFR}) with matrices $A_0=A^0$, $E_0=\begin{bmatrix}
      A^1 & A^2 & E^1 & E^2
    \end{bmatrix}$, $E_1=E^0$, $C_1=C^0$, $F_{10}=\begin{bmatrix}
      C^1 & C^2 & F^1 & F^2
    \end{bmatrix}$, $F_{11}=F^0$ and
\begin{equation}
  \begin{array}{lclclclclcl}
    C_0&=&\begin{bmatrix}
      I_n\\
      0_n\\
      0_{p\times n}\\
      0_{p\times n}
    \end{bmatrix},&&F_{00}&=&\begin{bmatrix}
      0_n & 0_n & 0_{n\times p} & 0_{n\times p}\\
      I_n & 0_n & 0_{n\times p} & 0_{n\times p}\\
      0_{p\times n} & 0_{p\times n} & 0_p & 0_p\\
      0_{p\times n} & 0_{p\times n} & I_p & 0_{p\times p}
    \end{bmatrix},&&F_{01}&=&\begin{bmatrix}
       0_{n\times p}\\
       0_{n\times p}\\
       I_p\\
       0_p
    \end{bmatrix}.
  \end{array}
\end{equation}
The matrices of the LFT form (\ref{eq:LFR2}) are given by $\bar{F}_{00}=F_{00}$, $\bar{C}_{00}=C_{0}$, $\bar{F}_{01}=F_{01}$, $\tilde{E}_0=\begin{bmatrix}
      A^{1T} & A^{2T} & C^{1T} & C^{2T}
    \end{bmatrix}$ and $\tilde{F}_{10}=\begin{bmatrix}
      E^{1T} & E^{2T} & F^{1T} & F^{2T}
    \end{bmatrix}$.

\begin{table}
\centering
  \begin{tabular}{c|c|c||c|c|}
    $\varphi_1(\delta)$ & $\varphi_2(\delta)$ & constraints & computed $ L_1$-gain & time [sec]\\
    \hline
    \hline
    $\varphi_1^0$ & $\varphi_2^0$ & $\varphi_1^0\ge0$, $\varphi_1^0+\varphi_2^0\ge0$ & 133.95 & 2.7844\\
     $\varphi_1^1\delta$& $\varphi_2^0$ & $\varphi_1^1=-\varphi_2^0$ & 133.95 & 3.829\\
    $\varphi_1^1\delta+\varphi_1^2\delta^2$ & $\varphi_2^0+\varphi_2^1\delta$ & $\varphi_1^1=-\varphi_2^0$, $\varphi_1^2=-\varphi_2^1$ & 94.167 & 4.2758\\
  \end{tabular}
\caption{$ L_1$-gain computation of the transfer $w_1\to z_1$ using Theorem \ref{th:robL1} -- Exact $ L_1$-gain: 92.8358}\label{tab:1}
\end{table}

\begin{table}
\centering
  \begin{tabular}{c|c|c||c|c|}
    $\varphi_1(\delta)$ & $\varphi_2(\delta)$ & constraints & computed $ L_1$-gain & time [sec]\\
    \hline
    \hline
    $\varphi_1^0$ & $\varphi_2^0$ & $\varphi_1^0\ge0$, $\varphi_1^0+\varphi_2^0\ge0$ & 86.195 & 0.68989\\
     $\varphi_1^1\delta$& $\varphi_2^0$ & $\varphi_1^1=-\varphi_2^0$ & 86.195 & 1.4629\\
    $\varphi_1^1\delta+\varphi_1^2\delta^2$ & $\varphi_2^0+\varphi_2^1\delta$ & $\varphi_1^1=-\varphi_2^0$, $\varphi_1^2=-\varphi_2^1$ & 82.025 & 1.7509\\
  \end{tabular}
\caption{$ L_\infty$-gain computation of the transfer $w_1\to z_1$ using Theorem \ref{th:robLinf} -- Exact $ L_\infty$-gain: 82.0249}\label{tab:inf}
\end{table}

Using Theorem \ref{th:robL1} and Theorem \ref{th:robLinf} (together with the complexity reduction technique of Section \ref{sec:solvingunc}) and different forms for the scalings $\varphi_1$ and $\varphi_2$, we obtain the results gathered in Tables \ref{tab:1} and \ref{tab:inf}. We can see that the $ L_1$-gain is not very well estimated compared to the $ L_\infty$-gain for parameter independent scalings. Using scalings of degree two considerably reduces the conservatism of the approach. Indeed, we are able to estimate accurately the $ L_\infty$-gain while little conservatism persists for the $ L_1$-gain. It is important to note that the above numerical results also hold in the case of time-varying parameters (see Section \ref{ex:3e}).

\section{Conclusion}

By relying on the fact that the $L_1$-gain and $L_\infty$-gain are related to the static-gain matrix of the positive system, a linear programming approach has been proposed to compute them under the assumption that the input signals are nonnegative. Interestingly, the results turn out to be also valid for inputs and states with no definite sign. Based on this result, stabilization conditions with performance constraints have been obtained for different classes of controllers. Using then the newly introduced Integral Linear Constraints, robust stability analysis results are provided together with generic robustness results, again relying on static-gains of positive systems. It is shown that the results are nonconservative in the case of LTI positive uncertainties with fixed static-gain matrix. Stabilization results are finally obtained and expressed as (robust) linear programming problems. An exact relaxation scheme based on Handelman's Theorem is proposed in order to obtain finite-dimensional feasibility problems. Several examples illustrate the approach.

\section*{Acknowledgments}

This work has been supported by the ACCESS and RICSNET projects, KTH, Stockholm, Sweden.

\bibliographystyle{IEEEtran}


\end{document}